\documentclass{article}

\usepackage{arxiv}
\usepackage[T1]{fontenc}
\usepackage[utf8]{inputenc}
\usepackage{hyperref}       
\usepackage{booktabs}       
\usepackage{microtype}      
\usepackage{amsmath,amsthm}
\usepackage{graphicx}
\usepackage{diagbox}
\usepackage{algorithm}
\usepackage{algpseudocode}
\usepackage{url}

\theoremstyle{plain}
\newtheorem{theorem}{Theorem}[section]
\newtheorem{lemma}[theorem]{Lemma}

\newtheorem{proposition}[theorem]{Proposition}

\theoremstyle{definition}

\theoremstyle{remark}
\newtheorem{remark}{Remark}

	\title{KALMAN FILTERING WITH CENSORED MEASUREMENTS}
	
	\author
	{
		Kostas~Loumponias
		\And
		George~Tsaklidis 
		\thanks{K. Loumponias and G. Tsaklidis are with the Department of Mathematics, Aristotle University of Thessaloniki, GR-54124, Thessaloniki, Greece (e-mail: kostikasl@math.auth.gr; tsaklidi@math.auth.gr) }
	}
	
\begin{document}
	
	\maketitle

	\begin{abstract}
		
		This paper concerns Kalman filtering when the measurements of the process are censored. The censored measurements are addressed by the Tobit model of Type I and are one-dimensional with two censoring limits, while the (hidden) state vectors  are multidimensional. For this model, Bayesian estimates for the  state vectors are provided  through a recursive algorithm of Kalman filtering type. Experiments are presented to illustrate the effectiveness and applicability of the  algorithm. The experiments show that the proposed method outperforms other filtering methodologies in minimizing the computational cost as well as the overall Root Mean Square Error (RMSE) for synthetic and real data sets.
		
	\end{abstract}
	

	\section{Introduction}
	
	Kalman filter (KF) \cite{kalman1960new} has been the subject of extensive research and application, particularly in the area of object tracking and vehicle navigation. The KF algorithm provides optimal estimates for hidden state vectors under the assumption that the  measurements given the state vectors are normally distributed and the corresponding state-space model is linear. However, in many real life problems, the state-space model is non-linear, therefore, the KF process has a poor performance. Many methods have been proposed in order to overcome these drawbacks of the standard KF, such as the Extended Kalman Filter (EKF) \cite{la1996design}, the Unscented Kalman Filter (UKF) \cite{larsen2011unscented},\cite{gustafsson2012some} etc. 
	
	A kind of non-linearity in state-space models is due to  censorship in the measurements  \cite{moore1956estimation}, \cite{hampshire1992tobit}, where methods such as EKF and UKF cannot cope in optimal way with the censored measurements. In what follows we will deal with this kind of non-linear state-space models, i.e., models with censored measurements. The use of statistics for censored data in filtering problems has received an increased attention over the last years \cite{ibarz2005kalman},\cite{allik2014tobit}. In \cite{ibarz2005kalman}, the censored measurements are treated as missing measurements, thus, only the state prediction (an a priori estimation) is utilized when a measurement is censored.  
	
	In \cite{allik2014tobit}, \cite{allik2016tobit}, the Tobit Kalman Filter (TKF) was proposed in order to estimate recursively the state vector,  given the censored measurement. The censored measurements are addressed by the Tobit model of Type I with two censoring limits \cite{tobin1958estimation}. TKF provides unbiased and recursive estimates of the state vectors as a linear combination of the a priori state vector estimation and the associated censored measurement, by taking into account the censoring limits. Furthermore, the TKF process is completely recursive and computationally inexpensive, thus constituting a perfect candidate for  investigating real-time applications. Nevertheless, since by the standard TKF algorithm no calculation of the exact covariance  matrix of the censored measurements is carried out, it provides non-optimal estimates \cite{allik2016tobit}. 
	
	In \cite{loumponias2018adaptive}, an online and real time multi-object tracking (MOT) algorithm based on censored measurements is presented. More specifically, the authors utilize the Adaptive Tobit Kalman Filter (ATKF) in order to estimate the position of the objects. The ATKF process is based on the same framework as TKF, while the two methods have two crucial differences. ATKF provides 1) the exact estimation of the variance of the censored measurements and 2) adaptive censoring limits at each time step, compared to TKF. In \cite{loumponias2016},\cite{loumponias2016using}, an ATKF was utilized in order to filter spatial coordinates of human skeleton (captured by Kinect Camera \cite{obdrzalek2012accuracy}), however, there was not used the exact covariance  matrix of the censored measurements.  
	Fei Han et al. \cite{han2018improved} deal with TKF for a class of linear discrete-time systems with random parameters. The elements of the state-space matrices are allowed to be random variables in order to reflect  reality. Furthermore, they establish a novel weighting covariance formula to address the quadratic terms associated with the random matrices. The method they propose  copes with one censoring limit.

	The main contributions of this paper is  the establishment of a new Censored Kalman filter (CKF) based on the conditional distribution function of the state vector (exhibiting the hidden states) when the measurements are censored. 
	In accordance with other studies dealing with censored measurements \cite{allik2014estimation},\cite{allik2016tobit}, we do not derive estimates  as a linear combination of the state vector's a priori estimation and the censored measurement. More specifically, Bayesian estimates \cite{chen2003bayesian} are calculated when the measurements lie into the censored area. Furthermore, we cope with i) a multidimensional hidden state vector, ii) one-dimensional censored measurement, and iii) the interval censoring (Type I censoring) \cite{turnbull1976empirical}, where a data point belongs into a bounded interval determined by known lower and upper limit. 
	For that purpose we provide: (a) the estimation of the first and second moment of a multidimensional random vector, conditional on an one-dimensional censored normal  variable, given that their joint unconditional and uncensored distribution  is normal, and
	(b) an accurate calculation of the associated likelihood function given the censored measurements. The proposed method, CKF, upgrades the standard KF process only when the measurements lie into the censored area. The results show that CKF has a better performance than TKF and KF, and a very low computational cost. Furthermore, CKF can be utilized for multidimensional censored measurements, in the case where the coordinates of any measurements are uncorrelated.   
	
	The rest of the paper is organized as follows: In Section 2,  Bayesian state estimates conditional on censored measurements are calculated, and the associated CKF algorithm is presented in detail. In Section 3, experimental results are illustrated using artificial and real data (Multi-Object Tracking) to demonstrate the effectiveness and the applicability of the proposed filtering algorithm. Finally, in Section 4, concluding remarks are provided.
	
	\section{Censored Kalman Filtering}
	
	In this section we deal with the KF process  with censored measurements. First, we describe briefly the meaning of censored measurements and the vanilla KF. Next, we calculate in detail the Bayesian estimates for the states at (discrete) time $t$ given the measurements (either censored or uncensored) till time $t$, denoted by $ {y}_{1},  {y}_{2},...,{y}_{T}$, or briefly as $Y_t = \textbf{y}_{1:t}$. Finally, we provide two recursive algorithms to cope with one-dimensional and multidimensional censored measurements, respectively.   
	
	\subsection{Censored measurements}
	
	The KF process uses a series of measurements, $Y_t$, observed over time, containing statistical noise, in order to estimate the set of unknown state vectors, $\textbf{x}_{1}, \textbf{x}_{2},...,\textbf{x}_{T}$, denoted by $ \textbf{x}_{1:T}$. The standard  state-space model is given by the equations 
	\begin{equation}
	\textbf{x}_{t+1}=\textbf{A}\textbf{x}_{t} +\textbf{w}_t,
	\label{x}
	\end{equation}
	\begin{equation}
	\textbf{y}^*_{t}=\textbf{H}\textbf{x}_{t} +\textbf{v}_t,
	\label{y}
	\end{equation}
	where \textbf{A}, \textbf{H} are the transition and observation matrices, respectively, and $\textbf{w}_t\sim N(\textbf{0},\textbf{Q}_t)$ and $\textbf{v}_t\sim N(\textbf{0},\textbf{R}_t)$ stand for the normally distributed noises of the  process and the measurement, respectively. While KF provides optimal estimates for the linear state-space model (\ref{x})-(\ref{y}), it turns out that many real life applications are described by non-linear state-space models, and consequently KF cannot cope with them. We notice that in such models the non-linearity often arises from  censored measurements, which is the case we will deal with. 
	
	Censoring is a condition in which the value of a measurement or observation is only partially known or unknown. In this paper we deal only with  the case of partially known measurements. A type of that kind of censoring is the Interval Censoring, where all observations lie in a finite interval. In the case of Interval Censoring, the measurements of the censored state-space model are defined by the relations
	
	\begin{equation}
	y_{t}=\begin{cases} y^*_{t},&a<y^*_{t}<b\\
	a,&y^*_{t} \leq a\\
	b,&y^*_{t} \geq b
	\end{cases},
	\label{censored}
	\end{equation}
	\begin{equation}
	y^*_{t} = \textbf{H}\textbf{x}_t + v_t,
	\label{y_lat}
	\end{equation}
	where $y_{t}$ and $y^*_{t}$ stand for the censored and latent (uncensored) measurements, respectively, and $a$ and $b$ are the lower and upper censoring limits, respectively. It is clear by (\ref{censored}), that the censored measurement, $ y_{t} $, is not normally distributed, while, $y^*_t\sim N(m_{y^*_t},r^2_t)$. Therefore, it is necessary to improve (upgrade) the standard KF in order to deal with the censored measurements.
	
	\subsection{Recursive Bayesian Estimations for Censored Measurements}
	
	In what follows we will apply  Bayesian estimation for estimating the unknown probability distribution function (pdf) $p(\textbf{x}_t \mid Y_t)$,  recursively over time using the incoming measurements, $Y_t$. In the case where the variables involved are normally distributed and the state-space model is linear, as given by (\ref{x})-(\ref{y}),  the Bayesian filter becomes  the standard KF. Two assumptions are used
	to derive the recursive Bayesian filter: a) the states follow a first-order Markov process, i.e.,  
	\[p(\textbf{x}_t \mid \textbf{x}_{0:t-1})=p(\textbf{x}_t\mid \textbf{x}_{t-1})\]
	and b) any measurement at some time $t$ does not depend on the previous states (till time $t-1$), given the current state (at time $t$). Using Bayes rule we get that \cite{chen2003bayesian},
	
	\begin{equation}
	\label{bayes}
	p(\textbf{x}_t \mid Y_t) = p(\textbf{x}_t \mid Y_{t-1}) \frac{ p(y_t \mid \textbf{x}_t)}{p(y_t \mid Y_{t-1})}.
	\end{equation}
	The aim of Bayes filter is to provide posterior  estimates for the states -which are considered to be the conditional means-  $\hat{\textbf{x}}_t = E(\textbf{x}_t|Y_t)$, and for the covariance matrices, $Cov(\textbf{x}_t \mid Y_t)$, given the distribution (\ref{bayes}).
	
	In
	what follows, we deal with one-dimensional censored measurements and assume that the random vector $\textbf{x}_t $ given $Y_{t-1} $ is normally distributed. Next, we provide a lemma where the conditional censored pdf $ f^c_{\textbf{x}|y}(\textbf{x}_t|y_t) $ is calculated via the corresponding unconditional censored pdf $ f^c_{\textbf{x},y}(\textbf{x}_t,y_t) $. For that purpose we will use the following notations: $m(\textbf{x}_t) = m_{y^*_t} + \textbf{S}_{\textbf{x}_t,y^*_t}\textbf{S}^{-1}_{\textbf{x}_t} (\textbf{x}_t - m_{\textbf{x}_t}) $, $ s^2 = s^2_{y^*_t} - \textbf{S}_{\textbf{x}_t,y^*_t}\textbf{S}^{-1}_{\textbf{x}_t}\textbf{S}_{\textbf{x}_t,y^*_t}^T$, $ \delta $ is the delta Kronecker function, $f_{\textbf{x},y^*}(\textbf{x}_t,y_t)$ stands for the normal pdf, $ u_{(a,b)}(y_t) $ is a function taking the value $1$ when $ y_t $ belongs to the interval $(a,b)$ and $ 0 $ otherwise, $f_{\textbf{x}}(\textbf{x}_t)$ stands for the marginal normal distribution function of $\textbf{x}_t$, $\Phi$ is the cumulative distribution function of the standard normal distribution, $m_{y^*_t}$, $m_{\textbf{x}_t}$ are the means of $y^*_t$ and $\textbf{x}_t$, respectively, $ s^2_{y^*_t} $ is the variance of $y^*_t$, $\textbf{S}_{\textbf{x}_t} $ is the covariance matrix of $ \textbf{x}_t $, $\textbf{S}_{\textbf{x}_t,y_t^*} $ is the cross-covariance matrix of $ \textbf{x}_t $ and $ y^*_t $, $ a^* = \frac{a-m_{y^*}}{s_{y^*}} $ and $ b^* = \frac{b-m_{y^*}}{s_{y^*}} $. Then the following lemma holds:
	
	\begin{lemma} The conditional censored pdf $ f^c_{\textbf{x}|y}(\textbf{x}_t|y_t) $ can be written in the form
		\begin{align*}
		f^c_{\textbf{x}|y}(\textbf{x}_t|y_t) =& f_{\textbf{x}|y^*}(\textbf{x}_t|y_t=y)u_{(a,b)}(y_t) \\
		&+f_{\textbf{x}}(\textbf{x}_t)\frac{ \Phi\Big(\frac{a-m(\textbf{x}_t)}{s} \Big)} 
		{\Phi(a^*)}\delta(y_t - a) \\
		&+f_{\textbf{x}}(\textbf{x}_t)\frac{ \Phi\Big(\frac{b-m(\textbf{x}_t)}{s} \Big)} 
		{1-\Phi(b^*)}\delta(y_t - b).  
		\end{align*}
		\label{gpdf}
	\end{lemma}
	
	\begin{proof}
		Obviously $ f^c_{\textbf{x},y}(\textbf{x}_t,y_t) $ is given by 
		\begin{align*}
		f_{\textbf{x},y}^c(\textbf{x}_t,y_t) =& f_{\textbf{x},y^*}(\textbf{x}_t,y_t)u_{(a,b)}(y_t) \\
		&+ \int_{-\infty}^{a} f_{\textbf{x},y^*}(\textbf{x}_t, y^*) dy^* \delta(y_t - a) \\
		&+ \int_{b}^{+\infty} f_{\textbf{x},y^*}(\textbf{x}_t, y^*) dy^* \delta(y_t - b),  
		\end{align*}
		from which follows that
		
		\begin{equation}
		\begin{aligned}
		f_{\textbf{x},y}^c(\textbf{x}_t,y_t) =& f_{\textbf{x},y^*}(\textbf{x}_t,y_t)u_{(a,b)}(y_t) \\
		&+ f_{\textbf{x}}(\textbf{x}_t)\Phi\Big(\frac{a-m(\textbf{x}_t)}{s} \Big)\delta(y_t - a) \\
		&+ f_{\textbf{x}}(\textbf{x}_t)(1- \Phi\Big(\frac{b-m(\textbf{x}_t)}{s}\Big))\delta(y_t - b).  
		\end{aligned}
		\label{jpdf}
		\end{equation}
		
		In order to calculate the conditional censored pdf $ f^c_{\textbf{x}|y}(\textbf{x}_t|y_t) $ via the unconditional one given in (\ref{jpdf}), we distinguish three cases: 1) $y_t \in (a,b)$, 2) $y_t = a$ and 3) $y_t = b$. In the case where $ y_t $ lies into the uncensored region $ (a,b) $,  we have that $f^c_{\textbf{x}|y}(\textbf{x}_t|y_t=y)$ is normally distributed and, more specifically,
		\begin{equation}
		f^c_{\textbf{x}|y}(\textbf{x}_t|y_t=y) = f_{\textbf{x}|y^*}(\textbf{x}_t|y_t=y).
		\label{uncensored}
		\end{equation}
		
		In the case where $ y_t=a $,  or equivalently, $y^*_t \leq a$, it is derived by (\ref{jpdf}) that
		
		\begin{equation}
		f^c_{\textbf{x}|y}(\textbf{x}_t|y_t = a)  = f_{\textbf{x}}(\textbf{x}_t)\frac{ \Phi\Big(\frac{a-m(\textbf{x}_t)}{s} \Big)} 
		{\Phi(a^*)}.
		\label{cond}
		\end{equation}
		
		In the same way, it follows that 
		
		\[
		f^c_{\textbf{x}|y}(\textbf{x}_t|y_t = b)  = f_{\textbf{x}}(\textbf{x}_t)\frac{ 1-\Phi\Big(\frac{b-m(\textbf{x}_t)}{s} \Big)} 
		{1-\Phi(b^*)}.
		\] 
	\end{proof}
	
	We observe that (\ref{cond}) has the same form as (\ref{bayes}), and more specifically:
	\begin{itemize}
		\item $f_\textbf{x} ( \textbf{x}_t)$ stands for the a priori distribution $p(\textbf{x}_t |Y_{t-1}),$
		\item $\Phi\Big(\frac{a-m(\textbf{x}_t)}{s} \Big)$ stands for the probability $p(y_t=a|\textbf{x}_t)$ and  
		\item ${\Phi(a^*)}$ stands for the propability $p(y_t=a|Y_{t-1}).$
	\end{itemize}
	Then, the following proposition can be proved.
	
	\begin{proposition}
		For a normally distributed multivariate random variable
		$(\textbf{X}_t,Y^*_t)$ with mean vector $\textbf{m}=[\textbf{m}_\textbf{x}, m_y^*]^T \; and \; $covariance matrix 
		$ \textbf{S} =
		\begin{bmatrix}
		\textbf{S}_\textbf{x}  & \textbf{S}_{\textbf{x},y^*} \\
		\textbf{S}_{\textbf{x},y^*} & s^2_{y^*}
		\end{bmatrix}$,
		the following statements hold: 
		
		\begin{enumerate}
			\item $E(\textbf{x}_t|y^*_t\leq a) =  E(\textbf{x}_t|y_t = a) = \textbf{m}_\textbf{x} - \frac{\textbf{S}_{\textbf{x},y^*}}{s_{y^*}}\frac{\phi(a^*)}
			{\Phi(a^*)},$
			
			\item $Cov(\textbf{x}_t - \hat{\textbf{x}}_t|y_t=a) = \textbf{S}_{\textbf{x}_t} - \frac{\textbf{S}_{\textbf{x},y^*}\textbf{S}_{\textbf{x},y^*}^T}{s^2_{y^*}}\Big(a^*\frac{\phi(a^*)}
			{\Phi(a^*)} + \Big(\frac{\phi(a^*)}{\Phi(a^*)}\Big)^2\Big),$ 
		\end{enumerate}
		where $ a^* = \frac{a-m_{y^*}}{s_{y^*}} $ and $\phi({x})$ stands for the pdf of standard normal distribution. 
		\label{prop_a}
	\end{proposition}
	
	\begin{proof}
		\begin{enumerate}
			\item
			We derive by (\ref{cond}) that
			
			\begin{align*}
			E(\textbf{x}_t|y_t = a) &= \frac{1}{\Phi(a^*)}\int_{-\infty}^{+\infty}\textbf{x}_t f_{\textbf{x}}(\textbf{x}_t)\Phi\Big(\frac{a-m(\textbf{x}_t)}{s} \Big)\textbf{dx}_t \\
			&=\frac{1}{\Phi(a^*)}\int_{-\infty}^{+\infty}\int_{-\infty}^{a}\textbf{x}_t f_{\textbf{x},y^*}(\textbf{x}_t,y^*_t)dy^*_t\textbf{dx}_t\\
			&=\frac{1}{\Phi(a^*)}\int_{-\infty}^{a}\Big(\int_{-\infty}^{+\infty}\textbf{x}_t f_{\textbf{x}|y^*}(\textbf{x}_t|y^*_t )\textbf{dx}_t \Big)f_{y^*}(y^*_t) dy^*_t\\
			&=\frac{1}{\Phi(a^*)}\int_{-\infty}^{a}E(\textbf{x}_t|y^*_t)f_{y^*}(y^*_t) dy^*_t\\
			&=\frac{1}{\Phi(a^*)}\int_{-\infty}^{a} \Big(\textbf{m}_\textbf{x} + \frac{\textbf{S}_{\textbf{x},y^*}}{s^2_{y^*}}(y^*_t -m_{y^*}) \Big)f_{y^*}(y^*_t) dy^*_t\\
			&=\frac{1}{\Phi(a^*)}\Big[\textbf{m}_\textbf{x}\cdot\Phi(a^*) + \frac{\textbf{S}_{\textbf{x},y^*}}{s^2_{y^*}}\Big(\int_{-\infty}^{a} \mkern-18mu y^*_t f_{y^*}(y^*_t) dy^*_t -m_{y^*}\Phi(a^*) \Big)\Big]\\
			&= \frac{1}{\Phi(a^*)}\Big[\textbf{m}_\textbf{x}\Phi(a^*) +\frac{\textbf{S}_{\textbf{x},y^*}}{s^2_{y^*}}\big(m_{tr,y^*}\cdot\Phi(a^*) -m_{y^*}\cdot\Phi(a^*) \big)\Big]
			\end{align*}
			\begin{equation}
			=\textbf{m}_\textbf{x} + \frac{\textbf{S}_{\textbf{x},y^*}}{s^2_{y^*}}\cdot \big(  m_{tr,y^*} -m_{y^*} \big), \qquad
			\qquad \qquad \; \; \,
			\label{cens}
			\end{equation}
			
			where $ m_{tr,y^*} $ stands for the truncated mean of the r.v. $y^*_t$ in the  interval $(-\infty, a)$ \cite{wilhelm2012moments} and is equal with 
			\begin{equation}
			m_{tr,y^*} = {m}_{y^*} -s_{y^*}\cdot \frac{\phi(a^*)}{\Phi(a^*)}.
			\label{trunc}
			\end{equation}
			Then we get by (\ref{cens}) and (\ref{trunc}) that
			\begin{equation}
			E(\textbf{x}_t|y_t = a) = \textbf{m}_\textbf{x} - \frac{ \textbf{S}_{\textbf{x},y^*} }{s_{y^*}} \cdot \frac{\phi(a^*)}{\Phi(a^*)}.
			\label{mean_x}
			\end{equation}
			
			\item We have that \[Cov(\textbf{x}_t - \hat{\textbf{x}}_t|y_t=a) = E(\textbf{x}_t \textbf{x}^T_t| y_t = a) -E(\textbf{x}_t|y_t = a)E(\textbf{x}_t|y_t = a)^T,\]
			where the second term has been evaluated in the first part of the theorem. Then,
			
			\begin{align*}
			E(\textbf{x}_t \textbf{x}^T_t| y_t =a) 
			=&\frac{1}{\Phi(a^*)}\int_{-\infty}^{+\infty}\textbf{x}_t\textbf{x}_t^T f_{\textbf{x}}(\textbf{x}_t)\Phi\Big(\frac{a-m(\textbf{x}_t)}{s} \Big)\textbf{dx}_t \\ 
			=&\frac{1}{\Phi(a^*)}\int_{-\infty}^{+\infty}\int_{-\infty}^{a}\textbf{x}_t\textbf{x}_t^T f_{\textbf{x},y^*}(\textbf{x}_t,y^*_t)dy^*_t\textbf{dx}_t\\
			=&\frac{1}{\Phi(a^*)}\int_{-\infty}^{a}\Big(\int_{-\infty}^{+\infty}\textbf{x}_t\textbf{x}_t^T f_{\textbf{x}|y^*}(\textbf{x}_t|y^*_t )\textbf{dx}_t \Big)f_{y^*}(y^*_t) dy^*_t\\
			=&\frac{1}{\Phi(a^*)}\int_{-\infty}^{a}E(\textbf{x}_t\textbf{x}_t^T|y^*_t)f_{y^*}(y^*_t) dy^*_t\\
			=&\frac{1}{\Phi(a^*)}\int_{-\infty}^{a} 
			\Big(\textbf{m}_\textbf{x}\textbf{m}^T_\textbf{x} + \textbf{S}_{\textbf{x}} 
			-\frac{\textbf{S}_{\textbf{x},y^*}\textbf{S}^T_{\textbf{x},y^*}}{s^2_{y^*}}\\ 
			&+\textbf{m}_\textbf{x}\frac{\textbf{S}^T_{\textbf{x},y^*}}{s^2_{y^*}}(y^*_t -m_{y^*}) +\textbf{S}_{\textbf{x},y^*}\frac{\textbf{m}^T_\textbf{x}}{s^2_{y^*}}(y^*_t -m_{y^*})\\
			&+\frac{\textbf{S}_{\textbf{x},y^*}\textbf{S}^T_{\textbf{x},y^*}}{(s^2_{y^*})^2}(y^{*2}_t -2m_{y*}y^*_t +m^2_{y*}) \Big)
			f_{y^*}(y^*_t) dy^*_t\\
			=&\textbf{m}_\textbf{x}\textbf{m}^T_\textbf{x} + \textbf{S}_{\textbf{x}} 
			-\frac{\textbf{S}_{\textbf{x},y^*}\textbf{S}^T_{\textbf{x},y^*}}{s^2_{y^*}}
			+\textbf{m}_\textbf{x}\frac{\textbf{S}^T_{\textbf{x},y^*}}{s^2_{y^*}}(m_{tr,y^*}-m_{y^*})\\
			&+\frac{\textbf{S}_{\textbf{x},y^*}\textbf{S}^T_{\textbf{x},y^*}}{(s^2_{y^*})^2}(m_{tr,y^{2*}} - 2m_{y^*}m_{tr,y^*} + m^2_{y^*})
			\end{align*}
			\begin{equation}
			+\textbf{S}_{\textbf{x},y^*}\frac{\textbf{m}^T_\textbf{x}}{s^2_{y^*}}(m_{tr,y^*}-m_{y^*}), \qquad \quad \,
			\label{joint_mean}
			\end{equation}
			
		\end{enumerate}
		where $m_{tr,y^{2*}}$ is the truncated second moment of  $y^*$ in the interval $(-\infty, a)$  \cite{wilhelm2012moments} and is given by
		\begin{equation}
		m_{tr,y^{2*}} = m^2_{y^*} - s^2_{y^*}a^*\frac{\phi(a^*)}{\Phi(a^*)} + m^2_{y^*} - 2m^2_{y^*}s_{y^*}\frac{\phi(a^*)}{\Phi(a^*)}.
		\label{sec_truncated}
		\end{equation}
		Thus, (\ref{joint_mean}) can be written by means of (\ref{trunc}) and (\ref{sec_truncated}) as 
		\begin{align*}
		E(\textbf{x}_t \textbf{x}^T_t| y_t = a) 
		=& \textbf{m}_\textbf{x}\textbf{m}^T_\textbf{x} + \textbf{S}_{\textbf{x}} 
		-\frac{\textbf{S}_{\textbf{x},y^*}\textbf{S}^T_{\textbf{x},y^*}}{s^2_{y^*}}
		-\textbf{m}_\textbf{x}\frac{\textbf{S}^T_{\textbf{x},y^*}}{s_{y^*}}\frac{\phi(a^*)}{\Phi(a^*)}   
		\end{align*}
		\begin{equation}
		\qquad \qquad \qquad\qquad \qquad \quad -\textbf{S}_{\textbf{x},y^*}\frac{\textbf{m}^T_\textbf{x}}{s_{y^*}}\frac{\phi(a^*)}{\Phi(a^*)}
		+\frac{\textbf{S}_{\textbf{x},y^*}\textbf{S}^T_{\textbf{x},y^*}}{s^2_{y^*}}\Big(1-a^*\frac{\phi(a^*)}{\Phi(a^*)}\Big).
		\label{joint}
		\end{equation}
		
		Then, we get by  (\ref{mean_x}) and (\ref{joint}) that
		\begin{equation}
		Cov(\textbf{x}_t - \hat{\textbf{x}}_t|y_t=a) = \textbf{S}_{\textbf{x}} - \textbf{S}_{\textbf{x},y^*}(s^2_{y^*})^{-1}
		\Big(a^*\frac{\phi(a^*)}{\Phi(a^*)} + \Big(\frac{\phi(a^*)}{\Phi(a^*)}\Big)^2   \Big)\textbf{S}^T_{\textbf{x},y^*}.
		\label{cov_x}
		\end{equation} 
	\end{proof}
	
	In the same way as presented in Proposition \ref{prop_a}, it can be proved that: 
	
	\begin{proposition}
		For a normally distributed multivariate random variable $(\textbf{X}_t,Y^*_t)$  with mean vector $\textbf{m}=[\textbf{m}_\textbf{x}, m_y^*]^T \; and \; $covariance matrix 
		$ \textbf{S} =
		\begin{bmatrix}
		\textbf{S}_\textbf{x}  & \textbf{S}_{\textbf{x},y^*} \\
		\textbf{S}_{\textbf{x},y^*} & s^2_{y^*}
		\end{bmatrix}$,
		the following statements hold:
		\begin{enumerate}
			\item $E(\textbf{x}_t|y^*_t\geq b) =  E(\textbf{x}_t|y_t = b) = \textbf{m}_\textbf{x} + \frac{\textbf{S}_{\textbf{x},y^*}}{s_{y^*}}\frac{\phi(b^*)}
			{1-\Phi(b^*)},$
			
			\item $Cov(\textbf{x}_t - \hat{\textbf{x}}_t|y_t=b) = \textbf{S}_{\textbf{x}_t} - \frac{\textbf{S}_{\textbf{x},y^*}\textbf{S}_{\textbf{x},y^*}^T}{s^2_{y^*}}\Big(b^*\frac{\phi(b^*)}
			{\Phi(b^*)} - \Big(\frac{\phi(b^*)}{\Phi(b^*)}\Big)^2\Big),$ 
		\end{enumerate}
		where $ b^* = \frac{b-m_{y^*}}{s_{y^*}} $. 
		\label{prop_b}
	\end{proposition}
	
	We notice that the random vector $  (\textbf{x}_t|y_t = a) $ is not normally distributed; nevertheless, the normality could be accepted  for various values of the censoring limit $ a $ and covariance matrix, $\textbf{S}_{\textbf{x},y^*} $. More precisely, this normality condition can be accepted  if the value of a is high enough while the correlation coefficient $r_{x,y}$ is low. To illustrate this statement, we consider the following example.
	
	Let $X$ and $Y$ $\sim N(0,1)$, then, in Table \ref{KS_test} the results of K-S tests for various values of censoring limits $a$ and correlation coefficients $r_{x,y}$ are presented; the sample size of K-S test is $n=1000$. In particular, the values of $a$ and $r_{x,y}$ are considered in the intervals $\big[-3.00,2.95\big]$ and $\big[0.05,0.95\big]$ with steps $0.35$ and $0.10$, respectively. As  can be seen in Table \ref{KS_test}, if $r_{x_i,y} \leq 0.75$ then for every value of $a$ the null hypothesis  $H_0:$ $(X|y=a) \sim N(\mu, \sigma^2)$ cannot be rejected, while for very high values of $r_{x,y}$ the null hypothesis has to be rejected. Thus, for our example, if $r_{x_i,y} \leq 0.75$ , we can accept that the distribution function $f(\textbf{x}|y=a)$  be approximated by a normal distribution with mean vector (\ref{mean_x}) and covariance matrix (\ref{cov_x}). We can get analogous results for the pdf $f(\textbf{x}|y=b)$. 
	
	Next, in Fig. \ref{fig_fxy} the distribution function  $f(x|y=-1.60)$ for $r_{x,y}= 0.75$ and $0.85$, respectively, is presented. Concerning the normality, notice that the conditional pdf $f(x|y=-1.60)$  for $r_{x,y}= 0.85$ is not symmetric, while $f(x|y=-1.60)$ for $r_{x,y}= 0.75$  represents approximately a normal distribution. 
	
	\begin{table}[ht]
		\caption{ K-S tests for the  hypothesis, $H_0:$ $(X|y=a) \sim N(\mu, \sigma^2)$, where 0 and 1 represent  acceptance and non-acceptance of $H_0$, respectively}
		\begin{center}
			\begin{tabular}{|c|c|c|c|} 
				
				\hline
				
				\diagbox{$a$~}{${r}_{x,y}$~~} & \big[\textbf{0.05},\textbf{0.75}\big] & \textbf{0.85 } & \textbf{0.95}\\
				
				\hline 
				\rule{0pt}{3ex}    
				\big[\textbf{-3.00},\textbf{-2.30}\big] & 0 & 0 & 1\\
				
				\hline
				\rule{0pt}{3ex}   
				\big[\textbf{-1.95},\textbf{0.85}\big] & 0 & 1 & 1\\
				
				\hline
				\rule{0pt}{3ex} 
				\textbf{1.20} & 0 & 0 & 1\\
				
				\hline
				\rule{0pt}{3ex}   
				\big[\textbf{1.55}, \textbf{2.95}\big]&  0 & 0 & 0\\
				\hline
				
			\end{tabular}
		\end{center}
		
		\label{KS_test}
	\end{table}
	
	\begin{figure}
		\centering
		\includegraphics[width=11cm]{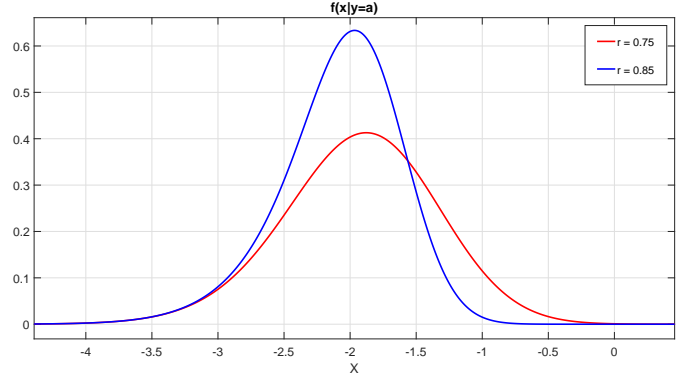}
		\caption{The distribution function $f(x|y=-1.60)$ for $r_{x,y}$= $0.75$ and $0.85$, respectively.}
		\label{fig_fxy}
	\end{figure}
	
	\subsection{The Proposed Model}
	
	The standard KF process consists of two stages: a) the predict stage and  b) the update stage. In the predict stage only the last state vector estimation is used, $\hat{\textbf{x}}_{t-1}$, in order to calculate the a priori estimation, $\hat{\textbf{x}}_{t}^-$ by (\ref{x}). The state vector at time ${t-1}$ given the  measurements up to time ${t-1}$, ${\textbf{x}}_{t-1}|Y_{t-1}$, is  normally distributed, and then by (\ref{x}), it is clear that ${\textbf{x}}_{t}|Y_{t-1}$ is normally distributed. In the censored KF process described by (\ref{censored})-(\ref{y_lat}), ${\textbf{x}}_{t-1}|Y_{t-1}$ is not normally distributed (see  Lemma \ref{gpdf}) when the last  measurement, $y_{t-1}$, belongs into the uncensored area; nevertheless, as can be seen in Table \ref{KS_test}, if the value of the corelation coefficient $r_{x_i,y}$ is not high, it can be accepted that $f^c_{\textbf{x}|y}({\textbf{x}}_{t-1}|y_{t-1}=a)$ is (approximately) normal. Therefore, as in the vanilla KF process,  the a priori state estimation and the corresponding  covariance matrix of the error of the a priori estimation, $ \textbf{P}_t^-$, are given by 
	
	\begin{equation}
	\hat{\textbf{x}}^-_t = \textbf{A}\hat{\textbf{x}}_{t-1}, 
	\label{x_apriori}
	\end{equation}
	\begin{equation}
	\textbf{P}^-_t = \textbf{A}\textbf{P}_{t-1}\textbf{A}^T+\textbf{Q}_t,
	\label{P_apriori}
	\end{equation}
	where $\hat{\textbf{x}}_{t-1}$ and $\textbf{P}_{t-1}$ are calculated at the previous step, $t-1$.
	
	At the next step the latent measurement, $y^*_t$, is used in order to update $\hat{\textbf{x}}_{t-1}$. In the case that $y^*_t$ belongs into the uncensored region $(a,b)$, we have that   $\textbf{x}_t|Y_t$ is normally distributed (see (\ref{uncensored})). Therefore, the a posteriori estimation, $\hat{\textbf{x}}_{t}$, and the corresponding error covariance matrix, $\textbf{P}_{t}$, can be calculated by the standard KF process in optimal way (i.e., unbiased and minimum variance estimation are provided). Thus, the KF algorithm has to be updated for the case where the measurements are censored; to that end, we utilize  Propositions \ref{prop_a} and \ref{prop_b}.  
	
	In the case where $y_t = a$, it is derived  by Proposition \ref{prop_a} and the state-space model (\ref{x}), (\ref{y}) that 
	\begin{equation}
	\textbf{x}_t|Y_{t-1} \sim N(\hat{\textbf{x}}_t^-,\textbf{P}_t^-)    
	\end{equation}
	and
	\begin{equation}
	{y}^*_t|Y_{t-1} \sim N(\textbf{H}\hat{\textbf{x}}_t^-,\textbf{H}\textbf{P}_t^-\textbf{H}^T + r_t). 
	\label{y_cond}
	\end{equation}
	Then, for $\textbf{m}_\textbf{x} =  \hat{\textbf{x}}^-_{t}$,  $\textbf{S}_{\textbf{x}_t} = \textbf{P}_t^-$, $s_{y^*} = (\textbf{H}\textbf{P}_t^-\textbf{H}^T + r^2_t)\in\Re^+ $ and $\textbf{S}_{\textbf{x},y^*} = \textbf{P}_t^-\textbf{H}^T$, we get by Proposition \ref{prop_a} that
	\begin{equation} 
	\hat{\textbf{x}}_{t} = \hat{\textbf{x}}^-_{t} - \frac{\textbf{P}_t^-\textbf{H}^T}{(\textbf{H}\textbf{P}_t^-\textbf{H}^T + r^2_t)^{1/2}}\frac{\phi(a^*)}{\Phi(a^*)}
	\label{x_a} 
	\end{equation}
	and 
	\begin{equation}
	\textbf{P}_t = \textbf{P}_t^- - \frac{\textbf{P}_t^-\textbf{H}^T\textbf{H}\textbf{P}_t^-}{\textbf{H}\textbf{P}_t^-\textbf{H}^T + r^2_t}\Big(a^*\frac{\phi(a^*)}
	{\Phi(a^*)} + \Big(\frac{\phi(a^*)}{\Phi(a^*)}\Big)^2\Big) 
	\label{P_a}.
	\end{equation}
	In the same way, when  $y_t=b$, it is derived that, 
	
	\begin{equation}
	\hat{\textbf{x}}_{t} = \hat{\textbf{x}}^-_{t} + \frac{\textbf{P}_t^-\textbf{H}^T}{(\textbf{H}\textbf{P}_t^-\textbf{H}^T + r^2_t)^{1/2}}\frac{\phi(b^*)}{1-\Phi(b^*)}
	\label{x_b}
	\end{equation}
	and 
	\begin{equation}
	\textbf{P}_t = \textbf{P}_t^- - \frac{\textbf{P}_t^-\textbf{H}^T\textbf{H}\textbf{P}_t^-}{\textbf{H}\textbf{P}_t^-\textbf{H}^T + r^2_t}\Big(b^*\frac{\phi(b^*)}
	{\Phi(b^*)} - \Big(\frac{\phi(b^*)}{\Phi(b^*)}\Big)^2\Big)
	\label{P_b},	
	\end{equation}
	where $a^* = \frac{a-\textbf{H}\hat{\textbf{x}}_t^-}{(\textbf{H}\textbf{P}_t^-\textbf{H}^T + r^2_t)^{1/2}}$ and 
	$b^* = \frac{b-\textbf{H}\hat{\textbf{x}}_t^-}{(\textbf{H}\textbf{P}_t^-\textbf{H}^T + r^2_t)^{1/2}}$.

	In dealing with real data and censored measurements, the measurement noise of the latent measurement $ y^*_t$, is usually unknown. In order to overcome this problem, we adopt the assumption that the latent measurement noise is normally distributed (\ref{y_lat}) (white noise with constant variance $r^2_t = r^2$). Then, we can estimate $r^2$ by means of the likelihood function of the censored measurements $ y^*_t$.
	The likelihood function for the censored measurements $ \{y_{t}\}_{t=1}^{T} $ given in (\ref{censored}), can be calculated by (\ref{y_cond}) as
	
	\begin{align}
	f(y_t|y_{t-1}) =& f_{y^*_t|y_{t-1}}(y_t)\cdot u_{(a,b)}(y_t) \nonumber\\ 
	&+ \Phi\Bigg(\frac{a-\textbf{H}\hat{\textbf{x}}^-_{t}}{(\textbf{H}\textbf{P}^-_{t}\textbf{H}^T+r^2)^{1/2}}\Bigg) \cdot \delta(y_t-a) \nonumber\\ 
	&+\Bigg(1-\Phi \Bigg(    \frac{b-\textbf{H}\hat{\textbf{x}}^-_{t}}{(\textbf{H}\textbf{P}^-_{t}\textbf{H}^T+r^2)^{1/2}} \Bigg)\Bigg) \cdot \delta(y_t-b) .
	\label{fy_con}
	\end{align}
	Then we get by (\ref{fy_con}), the following Lemma:
	
	\begin{lemma}
		The likelihood function of the censored normal distribution is given by 
		\begin{align}
		L(r^2|\textbf{y}) =& {\displaystyle\prod_{ a <~ y_{t} <~b} \frac{1}{(\textbf{H}\textbf{P}^-_{t}\textbf{H}^T+r^2)^{1/2}} \phi    \Bigg(\frac{y_{t}-\textbf{H}\hat{\textbf{x}}^-_{t}}{(\textbf{H}\textbf{P}^-_{t}\textbf{H}^T+r^2)^{1/2}} \Bigg) } \nonumber\\  
		&\times{\displaystyle\prod_{y_{t}=a}\Phi\Bigg(\frac{a-\textbf{H}\hat{\textbf{x}}^-_{i}}{(\textbf{H}\textbf{P}^-		_{t}\textbf{H}^T+r^2)^{1/2}}\Bigg)}  \nonumber\\ 
		&\times{\displaystyle\prod_{y_{t}=b}\Bigg(1-\Phi \Bigg(    \frac{b-\textbf{H}\hat{\textbf{x}}^-_{t}}{(\textbf{H}\textbf{P}^-_{t}\textbf{H}^T+r^2)^{1/2}} \Bigg)\Bigg)}.
		\label{likehood}
		\end{align}
		
	\end{lemma}
	
	\begin{remark}
		We notice that in  \cite{allik2014tobit}, the term  $\textbf{H}\textbf{P}^-_{t}\textbf{H}^T$ which appears in the denominators of (\ref{likehood}), is omitted (not considered); apparently this term is very important in order to estimate accurately the parameter $r^2$. The likelihood function given in \cite{allik2014tobit}, approximates (\ref{likehood}), in the cases where the observation matrix $\textbf{H}$ is equal to the identity matrix and $\textbf{P}_t^-$ is close to the null matrix.  
	\end{remark} 
	
	Next, the process of the proposed Censored Kalman Filter (CKF) is presented, where  $T$ in Algorithm \ref{algo} denotes the total number of the measurements. As can been stated by Algorithm \ref{algo}, the proposed method is  recursive and computationally inexpensive. More specifically, the proposed method has a similar computational burden to the standard KF \cite{tippett2003ensemble} and TKF \cite{allik2016tobit}, making it practical in computation-limited environments. 
	
	\begin{algorithm}[tbh] 
		\caption {Censored Kalman Filter}
		\begin{algorithmic}[1]
			\State $\textbf{x}_0\gets \textbf{0}_n$
			\State $\textbf{P}_0\gets \textbf{0}_{n \times n}$
			\For{\texttt{t=1:T}}
			\State  $\hat{\textbf{x}}_t^- \gets \textbf{A}\hat{\textbf{x}}_{t-1} $ 
			\State  $\textbf{P}_t^- \gets \textbf{A}\textbf{P}_{t-1}\textbf{A}^T +\textbf{Q}_t$
			\If {$y_t\in (a,b)$}
			\State \texttt{Utilize vanilla KF to compute $\hat{\textbf{x}}_t$ and $\textbf{P}_t$.}
			\ElsIf {$y_t = a$}
			\State $\hat{\textbf{x}}_t  \gets$  \texttt{Update using}  $ (\ref{x_a}) $ 
			\State  $\textbf{P}_t  \gets$ \texttt{Update using} $ (\ref{P_a})$
			\ElsIf {$y_t = b$}
			\State $\hat{\textbf{x}}_t  \gets$  \texttt{Update using} $ (\ref{x_b}) $ 
			\State $\textbf{P}_t \gets$  \texttt{Update using} $ (\ref{P_b})$
			\EndIf
			\EndFor
		\end{algorithmic}
		\label{algo}
	\end{algorithm}
	
	The Algorithm \ref{algo} can be generalised in the case of multidimensional measurements, $\textbf{y}^*_t = \{y^*_{t,i}\}^m_{i=1}$, when their coordinates are uncorrelated (Algorithm \ref{algo2}). The scope of this generalization in the multidimensional case is to provide a more computational efficient algorithm. The notation of Algorithm \ref{algo2} is as follow, $\textbf{a}$ and $\textbf{b}$ are the censored limits, $ \textbf{I}_m $ the unit matrix of size $ m $ and $\textbf{R}_t$ is the covariance matrix of measurement error, which in our case is diagonal since the coordinates of the measurement are uncorrelated.
	\begin{algorithm}[tbh] 
		\caption {Censored Kalman Filter Multidimensional Case}
		\begin{algorithmic}[1]
			\State $\textbf{x}_0\gets \textbf{0}_n$
			\State $\textbf{P}_0\gets \textbf{0}_{n \times n}$
			
			\For{\texttt{t=1:T}}
			\State $\textbf{G} \gets \textbf{I}_m$ 
			\State $\textbf{E}\gets \textbf{0}_m$
			\State  $\hat{\textbf{x}}_t^- \gets \textbf{A}\hat{\textbf{x}}_{t-1} $ 
			\State  $\textbf{P}_t^- \gets \textbf{A}\textbf{P}_{t-1}\textbf{A}^T +\textbf{Q}_t$
			\State $\textbf{S}_1 \gets \textbf{P}_t^-\textbf{H}^T$
			\State $\textbf{S}_2 \gets \textbf{H}\textbf{P}_t^-\textbf{H}^T + \textbf{R}_t$
			\State $\textbf{K} \gets \textbf{S}_1\textbf{S}_2 $
			\For{\texttt{i=1:m}}
			\If {$y_{t,i}\in (a_i,b_i)$}
			\State $ E[i] \gets y_{t,i} - (\textbf{H}\textbf{x}^-_t)_i$ 
			\ElsIf {$y_{t,i} = a_i$}
			\State $ E[i] \gets -sqrt(S_2[i,i])\phi(a^*_i)/\Phi(a^*_i) $ 
			\State $ G[i,i] \gets a^*_i\phi(a^*_i)/\Phi(a^*_i) + \big(\phi(a^*_i)/\Phi(a^*_i)\big)^2 $
			\ElsIf {$y_{t,i} = b_i$}
			\State $ E[i] \gets sqrt(S_2[i,i])\phi(b^*_i)/(1-\Phi(b^*_i)) $
			\State $ G[i,i] \gets -b^*_i\phi(b^*_i)/\Phi(b^*_i) + \big(\phi(b^*_i)/\Phi(b^*_i)\big)^2 $
			\EndIf
			\EndFor
			\State $\hat{\textbf{x}}_t \gets \hat{\textbf{x}}_t^- +\textbf{K}\textbf{E}$
			\State $\textbf{P}_t \gets \big(\textbf{I}_n - \textbf{K}\textbf{G}\textbf{H}\big)\textbf{P}_t^-$
			\EndFor
		\end{algorithmic}
		\label{algo2}
	\end{algorithm}
	
	\section{Experiments}
	
	In this section we conduct two sets of experiments-simulations (the same experiments as in \cite{allik2016tobit} and \cite{loumponias2018adaptive}) to evaluate the performance of the proposed method,  
	CKF, in comparison to KF and TKF. The first simulation concerns a damping oscillator and a simple oscillator (without damping). In the state-space models of the two oscillators, the same process and measurement noise is added. The second simulation deals with the problem of MOT using video sequences with static camera from the MotChallenge 2015 training database \cite{2DMOT2015}.

	\subsection{Oscillators}
	
	In the  experimental sets, we present a motivating example of tracking a sinusoidal model by a KF, TKF and CKF, when the measurements are saturated. The state space equations have the form (\ref{x}), with 
	
	\[
	\textbf{A}=c \cdot
	\begin{bmatrix}
	cos(\omega)&-sin(\omega)\\
	sin(\omega)&\quad cos(\omega)
	\end{bmatrix},
	\]
	and
	\[\textbf{H}=
	\begin{bmatrix}
	1&0
	\end{bmatrix},\]
	where $ \omega=0.005\,2\pi$ and $c \in R$. This simulation shows a tracking ability with a known model and unknown disturbance that enters the system through $ \textbf{w}_k$. In this example, the disturbance $ \textbf{w}_k $ is normally distributed, $\textbf{w}_k\sim N(\textbf{0},\textbf{Q}) $, where
	\[
	\textbf{Q}=
	\begin{bmatrix}
	0.05^2&0\\
	0&0.05^2
	\end{bmatrix},
	\]
	and the measurement noise, $ v_k $, is normally distributed, $v_k\sim N(0, r^2)$, with $r^2 = 0.5$. The initial state vector is  $ \textbf{x}_0=[5\quad 0]^T $ with covariance matrix $ \textbf{P}_0 = \textbf{I}_{2} $, and the  censored limits are  $ a = -0.5 $ and $ b = 0.5 $. Then, by the above example we produce censored (saturated) measurements, $y_k$, where $k=1,2,...,1000$.
	
	In our first experiment, we set $c=0.999$ (damped oscillator) and repeat the above procedure 100 times in order for the results to be more reliable. It is worth noticing that in real applications with censored data, the variance of the noise measurement, $r^2$, is not available, since the measurements are partially known. 
	In the proposed method in order to estimate the state vectors $\textbf{x}_t$, they are utilised only the estimations of $r^2$ calculated through the likelihood function (\ref{likehood}). Then, it is derived that the average estimate of $r^2$ is 0.51 with standard deviation 0.07. Therefore, the estimations of $r^2$ are very close to the real value $r^2 =0.5$. 
	
	Next, the filters' root-mean-squared errors (RMSEs) for each of the 100 iterations are calculated. The means of the filters' RMSEs for the iterations are presented in Table \ref{tab:rmse_1}, where we provide separate RMSEs for the two  estimated coordinates of the state vector, $ \textbf{x}_k$. The results show that KF (red coloured in Fig. \ref{a_0999}) has a very low performance and fails to cope with censored measurements, since KF considers them to be non-censored. TKF (orange) improves the results of KF, however, it does not provide optimal estimates, since its estimates are a linear combination of the a piori estimation and the censored measurement. CKF (purple) exhibits the best performance due to the fact that it provides Bayesian estimates assuming that $\textbf{x}_t|y_{t-1} \sim N(\hat{\textbf{x}}_t^-,\textbf{P}_t^-)$ (Fig. \ref{a_0999}). The computational costs for CKF, KF and TKF are 1.875 s, 1.872 s and 1.978 s (CPU: i5-3380M), respectively. The computational costs of KF and CKF are almost the same, since their procedures are identical when a measurement is not censored.        
	
	\begin{table}[h!]
		\caption{The means of the RMSEs for the filters KF, TKF and CKF,  for c=0.999.}
		\renewcommand{\arraystretch}{1.3}
		\begin{center}
			\begin{tabular}{ |c|c|c| }
				\hline
				\textbf{Filter}  & \textbf{Mean RMSE of} $\hat{\textbf{x}}_1 $ &  \textbf{Mean RMSE of} $\hat{\textbf{x}}_2 $ \\
				\hline
				KF &  2.0320 & 2.0431 \\
				TKF & 0.4431 & 0.5480 \\
				\textbf{CKF} & \textbf{0.3749} & \textbf{0.4966} \\
				\hline
			\end{tabular}
		\end{center}
		
		\label{tab:rmse_1}
	\end{table}
	
	\begin{figure}[h!]
		\centering
		\includegraphics[width=15cm]{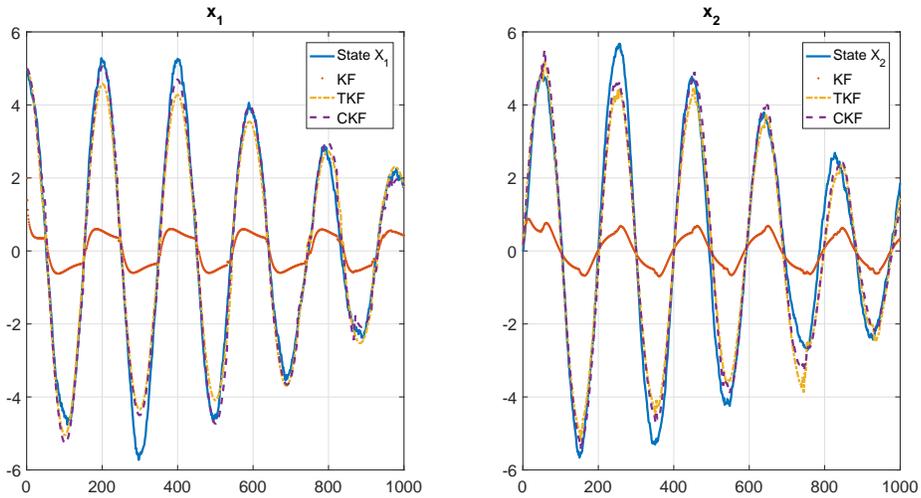}
		\caption{The state vector \textbf{x} and the estimations via KF, TKF and CKF,  for c=0.999.}
		\label{a_0999}
	\end{figure}
	
	In the second experiment, we set $c=1$ and repeat 100 times the evaluation procedure in the same way as in the first experiment. In Table \ref{tab:rmse_2}  the means of filters' RMSEs for the 100 iterations are presented.  It can be observed that KF has a poor performance, since is does not take into account censored measurements, while our proposed model, CKF, outperforms KF and TKF as in the first experiment (Fig. \ref{a_1}). 
	
	\begin{table}[h!]
		\caption{The means of the RMSEs for the filters KF, TKF and CKF, for c=1.}
		\renewcommand{\arraystretch}{1.3}
		\begin{center}
			\begin{tabular}{ |c|c|c| }
				\hline
				\textbf{Filter}  & \textbf{Mean RMSE of} $\hat{\textbf{x}}_1 $ &  \textbf{Mean RMSE of} $\hat{\textbf{x}}_2 $ \\
				\hline
				KF &  3.2149 & 3.2167 \\
				TKF & 0.6469 & 0.7202 \\
				\textbf{CKF} & \textbf{0.5489} & \textbf{0.6329} \\
				\hline
			\end{tabular}
		\end{center}
		
		\label{tab:rmse_2}
	\end{table}
	
	\begin{figure}[h!]
		\centering
		\includegraphics[width=15cm]{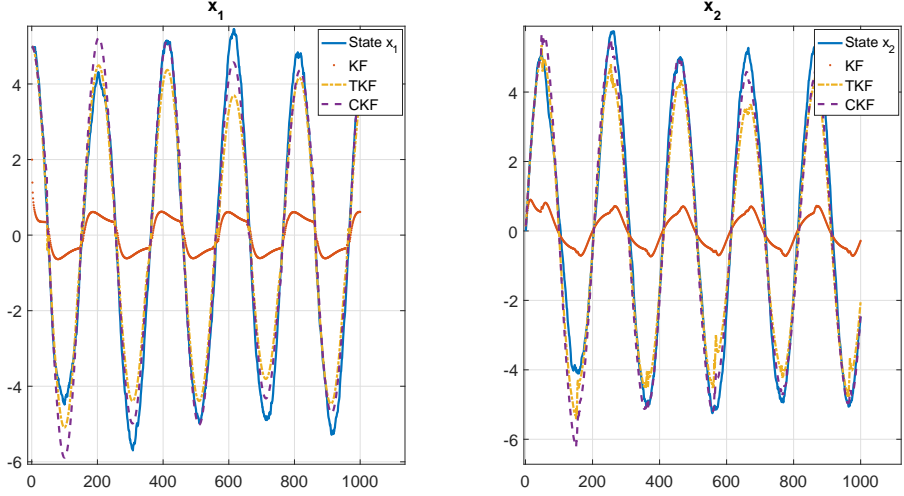}
		\caption{The state vector \textbf{x} and the estimations via KF, TKF and CKF, for c=1.}
		\label{a_1}
	\end{figure}
	
	Furthermore, we notice that the RMSEs are lower in the first experiment (damping oscillator) than in the second experiment (simple oscillator) due to the damping. More specifically, the probabilities of the measurements to lie in the uncensored region are higher in the case of the  damping oscillator than that one of the simple oscillator. Therefore, the RMSEs are increased for the three methods. However, we get by CKF  the lowest increase in RMSEs compared to KF and TKF, pointing out that the proposed method is more robust in the presence of censored measurements.

	\subsection{Multi-Object Tracking}
	
	Multi-Object Tracking (MOT) \cite{andriluka2008people} provides the estimation of the location of multiple objects in each time-frame of a video sequence while preserving the identity of the objects. One of the most well known frameworks that has been proposed to cope with MOT is "tracking-by-detection", where individual object detections are linked to form trajectories of the detected objects \cite{andriluka2008people}. The proposed method for MOT is based on three steps, as presented in \cite{loumponias2018adaptive}. In the third step the CKF process is used instead of the TKF. 
	
	Next, we outline  the proposed MOT process. 
	
	In the first step, any detections (i.e., detections of bounding-boxes coordinates with low confidence), are rejected in order to avoid false alarms. In addition, the non maximum suppression (NMS) algorithm \cite{neubeck2006efficient} is used in order to avoid multiple overlapping detections.  
	
	In the second step, the detections $\textbf{y}^*_t=[y^*_{t,1},y^*_{t,2},y^*_{t,3},y^*_{t,4}]$, concerning the four bounded-boxes coordinates at time frame $t$, are  associated with the predicted objects' positions, $\textbf{x}^-_t$, which are derived by (\ref{x_apriori}) and (\ref{P_apriori}). To that end, an assignment cost matrix is computed, as the intersection-over-union ($ IOU $) (see \cite{bewley2016simple}) between each detection and all predicted objects' positions from the existing targets. For the assignment between tracker and detection, the Hungarian algorithm \cite{jonker1986improving} is used. The  state vector
	$\textbf{x}_t=[x_{t,1},x_{t,2},x_{t,3},x_{t,4},\dot{x}_{t,1},\dot{x}_{t,2},\dot{x}_{t,3},\dot{x}_{t,4}]^T $ depicts the position coordinates of the object as long as the associated velocity coordinates at time frame $t$. The transition matrix $\textbf{A}$ of the CKF process takes on the form
	\begin{equation*}
	\textbf{A}=
	\begin{bmatrix}
	1&0&0&0&\frac{1}{fps}&0&0&0 \\
	0&1&0&0&0&\frac{1}{fps}&0&0 \\
	0&0&1&0&0&0&\frac{1}{fps}&0 \\
	0&0&0&1&0&0&0&\frac{1}{fps}\\
	0&0&0&0&1&0&0&0 \\
	0&0&0&0&0&1&0&0\\
	0&0&0&0&0&0&1&0\\
	0&0&0&0&0&0&0&1
	\end{bmatrix},
	\label{transition}
	\end{equation*}  
	where $ fps $ stands for the number of frames per second at each video sequence. The covariance matrix of the error process $ \textbf{Q}$ is taken to be, 
	\begin{equation*}
	\label{matrix_Q}
	\textbf{Q}=100*
	\begin{bmatrix} 
	1/4&0&0&0&1/2&0&0&0\\
	0&1/4&0&0&0&1/2&0&0\\
	0&0&1/4&0&0&0&1/2&0\\
	0&0&0&1/4&0&0&0&1/2\\
	1/2&0&0&0&1&0&0&0\\
	0&1/2&0&0&0&1&0&0\\
	0&0&1/2&0&0&0&1&0\\
	0&0&0&1/2&0&0&0&1
	\end{bmatrix}.
	\end{equation*} 	
	
	In the final step, the predicted object's position, $\textbf{x}^-_t$, is updated by taking into account the assigned censored detection, $\textbf{y}_t$, which is derived by (\ref{censored}) with adaptive censoring limits $\textbf{a}$ and $\textbf{b}$. More specifically, the censoring limits for each coordinate of the latent detection $\textbf{y}^*_t$ are given by  	
	\begin{equation}
	\label{limits_low}
	a_i = (\textbf{H}\hat{\textbf{x}}^-_t)_i - c, \qquad i=1,...,4 \;,
	\end{equation} 
	and	
	\begin{equation}
	\label{limits_up}
	b_i = (\textbf{H}\hat{\textbf{x}}^-_t)_i + c, \qquad i=1,...,4 \;,
	\end{equation}
	where $c$ is a positive constant and $\textbf{H}$ is the observation matrix which has the form
	\begin{equation*}\label{observation}
	\textbf{H}=
	\begin{bmatrix}
	1&0&0&0&0&0&0&0 \\
	0&1&0&0&0&0&0&0 \\
	0&0&1&0&0&0&0&0 \\
	0&0&0&1&0&0&0&0\\
	\end{bmatrix} \;.
	\end{equation*}
	The covariance matrix of the measurement error, $\textbf{R}_t$ is reasonable to be set inversely proportional to  a confidence value, $z_t$, which is provided for each detection; thus $ \textbf{R}_t $ is defined (experimentally) as,
	\begin{equation*}\label{R}
	\textbf{R}_t=
	\begin{bmatrix}
	81&0&0&0\\
	0&81&0&0\\
	0&0&81&0\\
	0&0&0&81\\
	\end{bmatrix}\cdot(1-\frac{z_t}{140})\;.
	\end{equation*}   
	Finally,  Algorithm \ref{algo2} is utilised to estimate the positions of the objects, i.e., the state vectors  $\hat{\textbf{x}}_t $. In what follows we let the parameter $c$ in (\ref{limits_low}) and (\ref{limits_up})  take different values. Apparently as  $c$ increases, then CKF and TKF converge to the standard KF; consequently, the performance of CKF and TKF is the same with the  performance KF for big values of $c$.
	
	Next, the following evaluation metrics concerning the performance, defined in \cite{bernardin2008evaluating}, are used:
	\begin{itemize}
		\item MOTA: Multi-object tracking accuracy. 
		\item MOTP: Multi-object tracking precision.
		\item FA: The average number of false alarms per frame.	
		\item FP: The total number of false positive detections.
		\item FN: The total number of false negative detections.
		\item ID sw: The total number of times an ID switches to a different previously tracked object.
		\item Frag: The total number of fragmentations where a track is interrupted by miss- detection.
		\item Hz: Processing speed (in frames per second excluding the detector) on the benchmark.
	\end{itemize}  
	The metric  MOTA  allows for objective comparison of the main characteristics of tracking systems, such as the accuracy in recognizing object configurations and the ability to consistently tracking objects over time.

	In Table \ref{mot_train} the proposed method,  CKF, is compared against  KF and TKF, on the training dataset, 2D MOT 2015 (static camera) \cite{2DMOT2015}.
	As shown in Table \ref{mot_train}, the proposed method for  $c=15$, outperforms TKF and KF in all metrics but one (Hz). More specifically, the highest MOTA value (equal to 32.8) is achieved by CKF for $c=15$ and the highest Hz value is achieved by KF, where the processing speed is equal to 448 fps. Furthermore, it is clear that the proposed method is $ 140\%-180\%$ faster than the TKF process for every value of $c$. This is due to the fact that the proposed method becomes the standard KF (which is time efficient) when the measurements lie into the uncensored region, while in the TKF process, the censored mean vector and covariance matrix are calculated for each measurement. Finally, we notice that MOTA decreases for the CKF and TKF, as  $c$ increases (e.g $c=25$). This result is expected, since both methods converge to the standard KF as the censoring limits (\ref{limits_low}), (\ref{limits_up})  increase.

	\begin{table*}[h]
		\renewcommand{\arraystretch}{1.3}
		\begin{center}
			\begin{tabular}{ |c|c|c|c|c|c|c|c| }
				\hline
				\ \textbf{Method}&\textbf{MOTA} $ \uparrow $  & \textbf{MOTP} $ \uparrow $ & \textbf{FA} $ \downarrow $& \textbf{FP} $ \downarrow $& \textbf{FN} $ \downarrow $&\textbf{ID sw} $ \downarrow $& \textbf{Hz} $ \uparrow $\\
				\hline
				\hline				
				\textbf{KF} & 31.9 & 72.1 & 0.70 & 2208 & 14370 & 155& \textbf{448}\\
				\hline
				\hline
				\textbf{CKF} c=10 & 31.9 & 71.9 & 0.69 & 2176 & 14716 & 142& 365\\
				\textbf{TKF} c=10 & 30.8 & 71.6 & 0.75 & 2350 & 14862 & 152& 151\\
				\hline
				\hline
				\textbf{CKF} c=15 & \textbf{32.8} & \textbf{72.1} & \textbf{0.67} & \textbf{2120} & \textbf{14615} &\textbf{ 136}& 384\\
				\textbf{TKF} c=15 & 32.1 & 71.9 & 0.70 & 2214 & 14682 & 141& 149\\
				\hline
				\hline
				\textbf{CKF} c=20 & 32.6 & \textbf{72.1} & 0.68 & 2122 & 14638 & 147& 406\\
				\textbf{TKF} c=20 & 32.6 & \textbf{72.1} & 0.68 & 2137 & 14633 & 138& 149\\
				\hline
				\hline
				\textbf{CKF} c=25 & 32.5 & \textbf{72.1} & 0.68 & 2129 & 14661 & 152& 416\\
				\textbf{TKF} c=25 & 32.5 & \textbf{72.1} & \textbf{0.67} & 2131 & 14656 & 148& 148\\
				\hline
				
			\end{tabular}
		\end{center}		
		\caption{Performance of KF, TKF and CKF on 2D MOT 2015 training sequences \cite{2DMOT2015} with static camera.}
		\label{mot_train}
	\end{table*}

	\section{Conclusion}
	
	In this paper a new framework in stochastic filtering for Tobit Type I censored measurements is established. In order to cope with these measurements, we propose the novel filtering method CKF,  which relies on Bayesian estimation. By taking into account that in many cases of non-linear state-space models, the exact Bayesian estimates cannot be calculated, numerical approximations can be provided. To that end, we assumed by means of the proposed methodology that the hidden state vector can be estimated through an  approximation of its pdf  by a normal distribution. This approximation can be accepted in several cases, based on the censoring limits and the correlation coefficient. Then, we calculated in detail the posterior estimation of a hidden state vector  and the corresponding covariance matrix of the coordinates'  error estimations, when  the measurements lie in the censored area. Furthermore, we calculated the associated likelihood function given the censored measurements.
	
	Next, we evaluated the proposed method against standard KF and TKF, by two different experiments. 
	In the first experiment-simulation, two noisy oscillators were used, a damping and a simple harmonic oscillator. As  expected, the KF process exhibited a very poor performance, since it is does not cope with non-linearity. Among the three approaches, CKF appeared to have the best performance in both oscillator simulations. Finally, for these simulations the variance of the (latent) measurement noise was estimated with high accuracy by means of the likelihood function; to that end, both uncensored and censored data were used.  
	In the second experiment, the dataset 2D MOT 2015, which deals with Multi-Object Tracking (MOT), was utilized. The results showed that the proposed method had a better performance in MOT accuracy than TKF and KF. It is worth noting that the CKF processing speed was higher up to $180\%$  than the corresponding speed of TKF.            
	
	Moreover, as a step further it would be interesting to extend the proposed method in multidimensional censored measurements with correlated coordinates, in order to describe  efficiently   real-life problems with censored measurements through non-linear state-space models.

	\bibliographystyle{unsrt}

\end{document}